\documentclass[10pt,twocolumn,twoside]{IEEEtran}
\usepackage{amsmath,amsfonts}
\usepackage{algorithmic}
\usepackage{array}
\usepackage{textcomp}
\usepackage{stfloats}
\usepackage{url}
\usepackage{verbatim}
\usepackage{graphicx}
\usepackage{balance}
\usepackage{xcolor}
\usepackage{subfig}
\usepackage{orcidlink}
\usepackage{booktabs} 
\usepackage{hyperref}
\hypersetup{hidelinks,
	colorlinks=true,    
	linkcolor=blue,     
	anchorcolor=blue,   
	citecolor=blue,     
	urlcolor=blue,       
	pdfstartview=Fit,
	breaklinks=true
}
\usepackage{cite}
\usepackage{amsmath,amssymb,amsfonts,amsthm}
\usepackage{algorithm}
\usepackage{flushend}
\newtheorem{theorem}{Theorem}
\newtheorem{lemma}{Lemma}
\newtheorem{ass}{Assumption}

\begin{document}
	\title{On Optimal Event-Triggered Distributed Control for Stochastic Multi-Agent Systems via Reinforcement Learning}
	
	\author{Ziming Wang\orcidlink{0000-0001-7000-9578}\textsuperscript{*}  \IEEEmembership{Student Member,~IEEE}, Bingbing Li\orcidlink{0009-0002-0787-5654}, Karl H. Johansson\orcidlink{0000-0001-9940-5929}  \IEEEmembership{Fellow,~IEEE} \\and Apostolos I. Rikos\orcidlink{0000-0002-8737-1984}  \IEEEmembership{Member,~IEEE}


           \thanks{* Corresponding author: Ziming Wang}
\thanks{Ziming Wang and Apostolos I. Rikos are with the Artificial Intelligence Thrust, The Hong Kong University of Science and Technology (Guangzhou), Guangzhou, China. Apostolos I. Rikos is also affiliated with the Department of Computer Science and Engineering, The Hong Kong University of Science and Technology, Clear Water Bay, Hong Kong. E-mails:{\tt~zwang216@connect.hkust-gz.edu.cn}, {\tt~apostolosr@hkust-gz.edu.cn}.}
\thanks{Bingbing Li is with the School of Economics \& Management, South China Normal University, Guangzhou, China. E-mail:{\tt~15320760579@163.com}.}
\thanks{Karl H. Johansson is with the Division of Decision and Control Systems, KTH Royal Institute of Technology, SE-100 44 Stockholm, Sweden. E-mail:{\tt~kallej@kth.se}.}}

	\maketitle
    
	\begin{abstract}
We propose a reinforcement learning (RL) based optimal distributed control algorithm for the multi-agent systems (MASs) with stochastic uncertainties. Unlike existing methods, during the optimized backstepping design process, we use the actor-critic-identifier structure. The actor neural network is used to reflect control behavior, the critic neural network works to evaluate control performance and the unknown stochastic uncertainties are handled by identifier neural network. 
Furthermore, a low-pass filter effectively suppresses problems stemming from non-affine nonlinear faults and a hybrid event-triggered control (ETC) strategy is proposed to reduce control frequency.
We analyze our algorithm’s operation, and we provide a Lyapunov-based stability proof that guarantees all errors are bounded, ensuring precise tracking between the leader and followers. We validate its correctness in a single-axis robotic manipulator simulation and finally, we compare against the non-optimal control algorithm highlighting our optimal control algorithm’s operational advantages.
	\end{abstract}
	

	\section{Introduction}
    \label{sec:introduction}
The relentless surge of intelligent systems has led to a growing need for effective control and coordination within stochastic multi-agent systems (MASs)~\cite{intro1,intro2,intro3}. A fundamental challenge in distributed coordination is leader-following consensus tracking problem~\cite{intro4,intro5}, which is to develop distributed control algorithms that can be used by a group of agents in order to reach agreement to a desired trajectory. The concept of optimal distributed control extends this objective to simultaneously achieve consensus while minimizing a cost function related to system performance. Knowledge of the optimal distributed control for stochastic MASs is of importance for various applications such as multi-robot systems~\cite{intro6}, smart grids~\cite{intro7}, unmanned vehicular systems~\cite{intro8}, and sensor networks~\cite{intro9}.


\subsection{Related Work}
The problem of stochastic MASs has received significant attention in the literature. A stochastic MAS is composed of a group of autonomous dynamic agents, which communicate with each other under a graph-based communication network. Among several studies on stochastic MASs~\cite{rela1,rela2,rela3}, consensus tracking control is one of the focal points. In general, two types of consensus problems have been studied, including leaderless consensus~\cite{rela1} and leader-following consensus~\cite{rela2,rela3}. It is worth mentioning that several high-order consensus control methods for the nonlinear stochastic MASs are developed in the recent works~\cite{rela4,rela5,rela6,rela7,jiajia1}. 
The work in \cite{rela4} introduced the idea of distributed logic into stochastic MASs. 
In~\cite{rela5}, a distributed consensus control method is developed for second-order stochastic MASs with external disturbances. 
In \cite{rela6} the authors consider the simultaneous presence of time delays and external disturbances, and achieves leader-following input-to-state mean-square consensus. 
Furthermore, \cite{rela7} applies the distributed stochastic control algorithm into a heterogeneous unmanned aerial vehicle swarm system. 

For decades, optimization has been a fundamental principle in modern control theory, originating from Bellman’s dynamic programming~\cite{bellman}. 
Most existing optimal control methods rely on solving the hamilton-jacobi-bellman (HJB) equation~\cite{rela8} or the hamiltonian equation~\cite{rela9}, although the HJB equation is generally difficult to solve due to its nonlinear nature. 
To overcome this difficulty, reinforcement learning (RL) based function approximation strategy~\cite{rela10} has been widely considered. 
Among other works, the actor-critic architecture is one of the most popular approaches for implementing RL algorithms, such as~\cite{rela11,rela12} for continuous systems and ~\cite{rela13,rela14} for discrete systems. 
In this architecture, the actor takes actions, and the critic evaluates them to improve future performance. 
However, most RL-based optimal control methods require full system dynamics, which is often impractical or infeasible. 
To relax this requirement, the actor-critic-identifier method was proposed in~\cite{rela15} to estimate unknown dynamics. 
Therefore, many recent studies focus on the actor-critic-identifier structure, as reported in~\cite{rela16, rela18}. Specifically, in~\cite{rela16}, they proposed an optimized backstepping control strategy to handle the dynamics in a nonlinear system, while~\cite{rela18} extended this approach to the MASs.


Early automatic control studies mainly employed continuous-time control strategies \cite{backstepping}, which often caused unnecessary resource consumption. 
This motivated the development of event-triggered control (ETC) strategies, like~\cite{compare, rela19, rela20, HETC1, HETC2}. 
Specifically, in~\cite{compare} and~\cite{rela19}, the fixed-threshold based ETC strategy was proposed to consistently reduce the control update frequency and save communication resources. The  work in \cite{rela20} proposed a self-triggered ETC strategy to predict the next triggering instant. 
Furthermore, the switching concept was introduced in~\cite{HETC1,HETC2}, where the triggering condition is divided into two components to integrate the advantages of different ETC strategies.

\subsection{Motivation}
Despite extensive studies on stochastic MASs, RL-based optimal control, and event-triggered control, several limitations remain. 
Existing stochastic consensus methods mainly focus on tracking control under stochastic disturbances, with greater emphasis on designing controllers to guarantee stability, convergence speed, and tracking performance~\cite{intro9,rela2,rela3}.
However, less attention has been paid to the optimization of the control process itself, especially the optimization of each virtual controller in the backstepping procedure. 
This issue becomes more challenging when unknown nonlinear dynamics and fault-induced uncertainties coexist~\cite{ass1,ass3}.

Therefore, it is necessary to develop an RL-based optimal distributed control framework that can simultaneously handle stochastic uncertainties and fault-induced uncertainties for MASs, reduce controller update frequency through a hybrid ETC strategy, and guarantee bounded tracking with Zeno-free behavior. Additionally, most existing studies analyze ETC results mainly by measuring the reduction in control signal triggering frequency, while the potential effects on tracking accuracy and control effort are often overlooked~\cite{lemma1,lemma2}. 
Hence, in this paper, we also aim to present a new analytical method that provides a more comprehensive quantitative evaluation of ETC performance.
It is necessary to develop an RL-based optimal distributed control framework that can simultaneously handle stochastic uncertainties and fault-induced uncertainties, reduce controller update frequency through a hybrid ETC strategy, and guarantee bounded tracking with Zeno-free behavior.
In addition, most existing studies analyze ETC results mainly by measuring the reduction in control signal triggering frequency, while the potential effects on tracking accuracy and control effort are often overlooked. 
Hence, a new analytical method is needed to provide a more comprehensive quantitative evaluation of ETC performance.

\subsection{Main Contributions}
Motivated by the aforementioned limitations in the existing literature, we propose an RL-based optimal distributed control algorithm for MASs.
To the authors' knowledge, our proposed framework provides a unified treatment of stochastic dynamics, non-specific fault-induced uncertainties, and hybrid event-triggered control, which has received limited attention in the existing literature.
The developed framework combines backstepping method, RL-based actor-critic-identifier structure, neural networks (NNs) and Lyapunov-based stability theory. 
The key contributions of this paper are summarized as follows:
\begin{itemize}
    \item We introduce an actor-critic-identifier structure-based optimal distributed control algorithm that simultaneously handles the unknown stochastic uncertainties and non-specific fault-induced uncertainties within a unified backstepping-based design (Section~\ref{sec3}). 
    
    \item We prove that all closed-loop errors within MASs are semi-globally uniformly ultimately bounded (SGUUB) (Theorem~\ref{theorem}).
    We also provide a hybrid ETC strategy by combining a static triggering condition and a dynamic triggering condition to achieve a balance between ensuring tracking performance and reducing triggering frequency. We guarantee the avoidance of Zeno behavior (Theorem~\ref{theorem}). 
    
    \item We demonstrate the effectiveness and applicability of the proposed algorithm in a practical single-axis robotic manipulator simulation (Section~\ref{sec4}). 
    Additionally, we compare the optimal control algorithm’s performance against the non-optimal control algorithms to highlight its advantages in tracking convergence speed and tracking accuracy. 
    Finally, drawing on knowledge from economics we provide a novel analysis of the event-triggered statistics, offering a new perspective for the analysis of ETC-related research (Section~\ref{sec4}).
\end{itemize}

\textbf{Paper Organization.} In Section~\ref{sec2}, we provide the notation, information regarding the agent communication, the stochastic MASs model, the hybrid ETC strategy and control objective. In Section~\ref{sec3}, we illustrate the controller design process and stability analysis. In Section~\ref{sec4}, we present our optimal control algorithm on a single-axis robotic manipulator simulation to validate the effectiveness.  Section~\ref{sec5} concludes the paper.
 
\section{Preliminaries and Problem Formulation}\label{sec2}
\textbf{Notation.} Let $\mathbb{R}$ denote the set of real numbers,  $\mathbb{R}^{n\times m}$ denote the set of $n\times m$ real matrices, and $\mathbb{Z}^+$ denote the set of positive integers.
        Define $\chi$ as a general variable. $\hat{\chi}$ means the estimation of $\chi^*$. $\tilde{\chi}=\hat{\chi}-\chi^*$ denotes the estimation error. $\Theta_{min}(\chi)$ denotes the minimum eigenvalue of the positive definite matrix $\chi$. $\Theta_{max}(\chi)$ denotes the maximum eigenvalue of the positive definite matrix $\chi$. $e$ denotes the Euler's number. $||\cdot||$ denotes the Euclidean norm for vectors or the induced two-norm for matrices.

\textbf{Communication Network.} Graph theory serves as a fundamental method for modeling communication interactions in the MASs among agents including $1$ leader and $N$ followers. The topology is defined as $\mathcal{G}=(\mathcal{V},\mathcal{E})$, where $\mathcal{V}=\{v_1,v_2,...,v_N\}$ denotes the nodes and $\mathcal{E}\subseteq \mathcal{V}\times\mathcal{V}$ denotes the edges (self edges excluded). $(v_i,v_j)\in \mathcal{E}$ indicates that agent $i$ can receive information from agent $j$. we define $\mathcal{A}=[a_{i,j}]$ as the adjacency matrix of $\mathcal{G}$, where $a_{i,j}=1$ indicates information transmission from agent ${j}$ to agent ${i}$, whereas $a_{i,j}=0$ indicates no communication interactions between the two agents. Then we define the Laplacian matrix $\mathcal{L}$ as $\mathcal{L}=\mathcal{D}-\mathcal{A}$, where $\mathcal{D}=\operatorname{diag}( {{d_1},{d_2},...,{d_N}})$ is the in-degree matrix of graph $\mathcal{G}$, ${d_i} = \sum\nolimits_{j = 1,j \ne i}^N {{a_{ij}}}$. And $\mathcal{B} = \operatorname{diag}({{b_1},{b_2},...,{b_N}})$ denotes the communication weights between the follower and the leader, where ${b_i} > 0$ if the leader can transfer information to follower $i$, if not, ${b_i} = 0$ holds.

\subsection{Stochastic Multi-agent Systems Model}
    Consider the MASs with $N$ followers labeled from $1$ to $N$ and one leader. Each follower is modeled by the following stochastic nonlinear dynamics:
        \begin{equation}
		\label{equ1}
		\begin{aligned}
			dx_{i,k}&=(x_{i,k+1}+h_{i,k}(\bar{x}_{i,k}))dt+f_{i,k}(\bar{x}_{i,k})dw_i, \\
			dx_{i,n}&=(u_i+h_{i,n}(\bar{x}_{i,n})+c_i\xi_{i}(\bar{x}_{i,n},u_i))dt+f_{i,n}(\bar{x}_{i,n})dw_i,\\
            y_i&=x_{i,1},
		\end{aligned}
	    \end{equation}
    where $i=1,2,...,N$ and $k=1,2,...,n-1$. $\bar{x}_{i,k}=[x_{i,1},x_{i,2},...,x_{i,n}]^T\in\mathbb{R}^n$ denotes the state vector of the $i$-th follower. $h_{i,k}\in\mathbb{R}$ and $f_{i,k}\in\mathbb{R}$ are unknown smooth Lipschitz continuous functions. The stochastic dynamics $w_i\in\mathbb{R}$ denotes the $r$-dimensional standard Wiener process defined on a probability space $(\Omega,\mathcal{F},\{\mathcal{F}_t\}_{t\geq t_0},\mathcal{P})$ with $\Omega$ denoting a sample space, $\mathcal{F}$ denoting a $\sigma$-field, $\{\mathcal{F}_t\}_{t\geq t_0}$ denoting a filtration and $\mathcal{P}$ denoting the probability measure. $u_i\in\mathbb{R}$ and $y_i\in\mathbb{R}$ denote the input and output of the $i$-th follower. $\xi_i\in\mathbb{R}$ denotes the non-specific fault-induced uncertainties with $c_i(t-t^*)$ denoting the temporal characteristics of faults occurring at an unknown time $t^*$. Notably, $c_i(t-t^*)=0$ holds when $t<t^*$ and $c_i(t-t^*)=1-\mathrm{e}^{-\mu\left(t-T_0\right)}$ otherwise. $\mu$ denotes the evolution rate of the unknown fault.

\subsection{Hybrid Event-triggered Control Strategy}
For many real-world scenarios, it is unrealistic to have continuous transmission of control signals. Applying the ETC strategy ensures that control signals are transmitted discretely only when necessary, which aligns more closely with practical scenarios~\cite{HETC1,HETC2}. Moreover, by regulating the triggering condition, control resources can be utilized more efficiently and the potential communication channel congestion can be avoided.

In this paper, we present the hybrid ETC strategy, containing a static triggering condition and a dynamic triggering condition. First, let the nominal time be denoted as $t$, the triggering instants as $t_s$. Define the updated controller in the hybrid ETC strategy as $w_i(t)$ and the final controller as $u_i(t)$. $u_i(t)=\kappa_i(t_s)$ in $t\in[t_s,t_{s+1})$. Define the hybrid gate as $\mathcal{H}_g$.  When $|u_i|\geq \mathcal{H}_g$, this strategy reduces trigger frequency at static-threshold intervals to improve efficiency and avoid large impulses. When $|u_i|\leq \mathcal{H}_g$, the dynamic-threshold value depends on the magnitude of $u_i$. This enables precise control when needed. Hence, this strategy not only achieves a reasonable update interval but also avoids excessively large impulses.

\subsection{Control Architecture and Problem Statement}
\textbf{Control Architecture.} The MAS consists of $N$ followers and one leader. Based on Fig.~\ref{topo}, the leader provides the desired trajectory $y_r\in\mathbb{R}$ to the followers through the leader-follower communication matrix $\mathcal{B}$, while the information exchange among followers is characterized by the graph topology $\mathcal{G}$ and the adjacency matrix $\mathcal{A}$. For each follower, the control input is computed using its own state information, the output information received from its neighboring followers, the available leader information, and the local control modules, including the actor-critic-identifier neural networks and the hybrid event-triggered strategy. Direct information exchange may occur between neighboring followers according to the prescribed communication topology and no central coordinator is required during the control process. The control architecture is therefore a graph-based distributed leader-following consensus framework.

\textbf{Problem Statement.} Consider the MASs including $N$ followers. The dynamics of each follower are governed by the second-order system~\eqref{equ1}. During operation, each follower is subject to the stochastic dynamics $w_i\in\mathbb{R}$ and the non-specific fault-induced uncertainties $\xi_i\in\mathbb{R}$. Under these conditions, the control objectives of this paper are formulated as follows:
\begin{itemize}
    \item Each follower is capable of tracking the leader's desired trajectory $y_r$ with accuracy. All tracking errors within the system~\eqref{equ1} are SGUUB.

    \item The hybrid ETC strategy can effectively reduce the triggering frequency of control signals and achieve a balance between control frequency conservation and tracking performance while strictly avoiding Zeno behavior (i.e., ruling out an infinite number of triggers in a finite time interval).

    \item By integrating the economic idea of benchmark normalization, we proposes a novel analytical framework for ETC and opens up new directions for subsequent ETC research.
\end{itemize}

\section{Distributed Consensus Controller Design and Stability Analysis}\label{sec3}
Following the graph-based distributed leader-following consensus control architecture described in Section~\ref{sec2}-C, we now introduce the distributed controller design process by utilizing backstepping method~\cite{HETC1}. Then we propose some assumptions and lemmas and outline the stability analysis according to the Lyapunov-based stability theory~\cite{HETC2}.

\subsection{Distributed Consensus Controller Design}
The design procedure on the $i$th follower contains $n$ steps by using backstepping method. Consider the MASs~\eqref{equ1}, we define the graph-based tracking error $z_{i,k}$ as
        \begin{equation}
		\label{equ2}
		\begin{aligned}
			z_{i,1}&=\sum_{j=1}^Np_{i,j}(y_i-y_j)+q_i(y_i-y_r), \\
			z_{i,k}&=x_{i,k}-u^{*}_{i,k-1},
		\end{aligned}
	    \end{equation}
        where $i=1,2,...,N$ and $k=2,...,n$. We have $u^{*}_{i,k-1}$ as the optimal virtual controller during step $k-1$ designed by backstepping method and RL. From the MASs~\eqref{equ1}, the systems require the recursive and systematic control design procedure because the control input is not matched with system nonlinearities. We introduce the backstepping design procedure to satisfy this requirement.

        \textbf{Step 1.} According to~\eqref{equ1} and~\eqref{equ2}, the derivative of $z_{i,1}$ is 
        \begin{equation}
		\label{equ3}
		\begin{aligned}
			\dot{z}_{i,1}=&[\eta_i(x_{i,2}+h_{i,1}(x_{i,1}))-\sum_{j=1}^Np_{i,j}(x_{i,2}+f_{i,1}(x_{i,1}))\\&-q_i\dot{y}_r]dt+[\eta_if_{i,1}(x_{i,1})-\sum_{j=1}^Np_{i,j}f_{i,1}(x_{i,1})]dw_i,
		\end{aligned}
	    \end{equation}
        where $\eta_i=\sum_{j=1}^Np_{i,j}+q_i$.

        To achieve the control objective of accurate tracking, we employ the actor-critic-identifier structure from reinforcement learning to optimize the control process, and we design a system performance index function $\mathcal{J}_{i,1}(z_{i,1})$ as follows: 
        \begin{equation}
		\label{equ4}
		\begin{aligned}
			\mathcal{J}_{i,1}(z_{i,1})=\int_t^\infty o_{i,1}(z_{i,1}(\chi),\bar{u}^{*}_{i,1}(\chi))d\chi,
		\end{aligned}
	\end{equation}
	where $o_{i,1}(z_{i,1},\bar{u}_{i,1})=z_{i,1}^2+\bar{u}^{*2}_{i,1}$ is the value function, and $\bar{u}^{*}_{i,1}\in\mathbb{R}$ is the virtual controller designed by backstepping method. Then we define $u^{*}_{i,1}\in\mathbb{R}$ as the optimal virtual controller. Given a compact set $\Omega$ that includes both the origin and the desired trajectory $y_r(t)$, The admissible control is represented by $\psi(\Omega)$. Then the optimal system performance index function is formulated as
    \begin{equation}
		\label{equ5}
		\begin{aligned}
			\mathcal{J}_{i,1}^{*}(z_{i,1})&=\int_t^\infty h_{i,1}(z_{i,1}(\chi),u_{i,1}^*(\chi))d\chi\\&=\min_{\bar{u}_{i,1}\in\psi(\Omega)}\int_t^\infty h_{i,1}(z_{i,1}(\chi),\bar{u}_{i,1}(\chi))d\chi.
		\end{aligned}
	\end{equation}

    To determine the optimal control action to be taken in any state, the HJB equation is derived by computing the time derivative on both sides of~\eqref{equ5} as follows
    \begin{equation}
		\label{equ6}
		\begin{aligned}
			&H_{i,1}(z_{i,1},u_{i,1}^{*},\mathcal{J}_{i,1}^{*})
			=z_{i,1}^{2}+u_{i,1}^{*2}+\mathcal{Z}_{i,1}=0,
	   \end{aligned}
	\end{equation}
    where
    \begin{equation}
		\label{equ7}
		\begin{aligned}
			\mathcal{Z}_{i,1}=&\frac{d\mathcal{J}_{i,1}^{*}}{dz_{i,1}}\times[\eta_i(x_{i,2}+h_{i,1})-\sum_{j=1}^Np_{i,j}(x_{i,2}+f_{i,1})\\&-q_i\dot{y}_r]+\frac{1}{2}\frac{d^2\mathcal{J}_{i,1}^{*}}{dz^2_{i,1}}\times(\eta_if_{i,1}-\sum_{j=1}^Np_{i,j}f_{i,1})^2.
	   \end{aligned}
	\end{equation}
    
    Based on the formulation mentioned above, the optimal virtual controller $u_{i,1}^*$ uniquely corresponds to the optimal performance index function~\eqref{equ5}. Therefore, it is necessary for this unique control solution to satisfy the HJB equation~\eqref{equ6}\eqref{equ7}. Then, the optimal virtual controller $u_{i,1}^*$ can be derived by solving equation $\partial H_{i,1}/\partial u_{i,1}^*=0$ as
    \begin{equation}
		\label{equ8}
		\begin{aligned}
		u_{i,k}^*=-\frac{\eta_i}{2}\frac{d\mathcal{J}_{i,1}^*}{dz_{i,1}}.
		\end{aligned}
	\end{equation}

    However, due to the strong nonlinearity of the term $({d\mathcal{J}_{i,1}^*}/{dz_{i,1}})$, the optimal virtual controller $u_{i,1}^*$ in~\eqref{equ8} is unattainable. In order to achieve the leader-following consensus tracking purpose, we introduce the consensus error term $(d\mathcal{J}_{i,1}^*/dz_{i,1})$ with a designed positive parameter $\zeta_{i,1}$ as
    \begin{equation}
		\label{equ9}
		\begin{aligned}
		\frac{d\mathcal{J}_{i,1}^*}{dz_{i,1}}=\frac{1}{\eta_i^2}(\zeta_{i,1}z_{i,1}+{J}_{i,1}^0+2h_{i,1}^*),
		\end{aligned}
	\end{equation}
    where
    \begin{equation}
		\label{equ10}
		\begin{aligned}
			{J}_{i,1}^0=& \eta_i^2\frac{d\mathcal{J}_{i,1}^*}{dz_{i,1}}-2\zeta_{i,1}z_{i,1}-2h_{i,1}^*,\\
            h_{i,1}^*=& -\sum_{j=1}^Np_{i,j}(x_{j,2}+h_{j,1})+z_{i,1}(\eta_if_{i,1}-\sum_{j=1}^Nf_{j,1})^4\\&+\eta_ih_{i,1}+\frac{3}{4}(\eta_iz_{i,1}+z_{i,1}^3)-q_i\dot{y}_r.
	   \end{aligned}
	\end{equation}

    Now we introduce radial basis function NNs~\cite{simu2,simu3} to approximate the unknown terms ${J}_{i,1}^0$ and $h_{i,1}^*$ as
        \begin{equation}
		\label{equ11}
		\begin{aligned}       \mathcal{J}_{i,1}^0&=W^{*T}_{\mathcal{J}_{i,1}}Q_{\mathcal{J}_{i,1}}+\rho_{\mathcal{J}_{i,1}},\\
        h_{i,1}^*&=W^{*T}_{h_{i,1}}Q_{h_{i,1}}+\rho_{h_{i,1}},
		\end{aligned}
	\end{equation}
    where $W^{*}_{\mathcal{J}_{i,1}}$ and $W^{*}_{h_{i,1}}$ are the ideal weights, $Q_{\mathcal{J}_{i,1}}$ and $Q_{h_{i,1}}$ are the basis function vectors, $\rho_{\mathcal{J}_{i,1}}$ and $\rho_{h_{i,1}}$ are the approximation errors. Then we design the following actor-critic-identifier NNs structure. With $\rho_{i,1}=\rho_{\mathcal{J}_{i,1}}+2\rho_{h_{i,1}}$, we design the optimal virtual controller $\hat{u}^*_{i,1}$ and the actor NNs $\hat{W}_{u_{i,1}}$ as
        \begin{equation}
		\label{equ12}
		\begin{aligned}
        \hat{u}^*_{i,1}&=-\frac{1}{\eta_i}(\zeta_{i,1}z_{i,1}+\hat{W}^{T}_{h_{i,1}}Q_{h_{i,1}}+\frac{1}{2}\hat{W}_{u_{i,1}}^TQ_{\mathcal{J}_{i,1}}),\\
        \dot{\hat{W}}_{u_{i,1}}&=-Q_{\mathcal{J}_{i,1}}^TQ_{\mathcal{J}_{i,1}}(\varphi_{u_{i,1}}(\hat{W}_{u_{i,1}}-\hat{W}_{c_{i,1}})+\varphi_{c_{i,1}}\hat{W}_{c_{i,1}}),
		\end{aligned}
	\end{equation}
    where $\varphi_{u_{i,1}}$, $\varphi_{c_{i,1}}$ are designed parameters that satisfy $\varphi_{u_{i,1}}>1/2$ and $\varphi_{u_{i,1}}>\varphi_{c_{i,1}}>\varphi_{u_{i,1}}/2$. Then we design the consensus error term $(d\mathcal{J}_{i,1}^*/dz_{i,1})$ and the critic NNs $\hat{W}_{c_{i,1}}$ for evaluating control performance as 
    \begin{equation}
		\label{equ13}
		\begin{aligned}
        \frac{d\mathcal{J}_{i,1}^*}{dz_{i,1}}&=\frac{1}{\eta_i^2}(2\zeta_{i,1}z_{i,1}+2\hat{W}^{T}_{h_{i,1}}Q_{h_{i,1}}+\hat{W}^T_{c_{i,1}}Q_{\mathcal{J}_{i,1}}),\\
        \dot{\hat{W}}_{c_{i,1}}&=-\varphi_{c_{i,1}}Q_{\mathcal{J}_{i,1}}^TQ_{\mathcal{J}_{i,1}}\hat{W}_{c_{i,1}}.
		\end{aligned}
	\end{equation}

    The identifier NNs $\hat{W}_{h_{i,1}}$ for estimating the unknown dynamic function $h^*_{i,1}$ are designed as
    \begin{equation}
		\label{equ14}
		\begin{aligned}
        \hat{h}^*_{i,1}&=\hat{W}^{T}_{h_{i,1}}Q_{h_{i,1}},\\
        \dot{\hat{W}}_{h_{i,1}}&=\delta_{i,1}(Q_{h_{i,1}}z^3_{i,1}-\sigma_{i,1}\hat{W}_{h_{i,1}}),
		\end{aligned}
	\end{equation}
    where $\delta_{i,1}$ and $\sigma_{i,1}$ are designed positive parameters.

    \textbf{Step k and Step n (2 $\leq$ k $\leq$ n-1).} According to~\eqref{equ1} and~\eqref{equ2}, the derivatives of graph-based tracking errors $z_{i,k}$ and $z_{i,n}$ are
    \begin{equation}
		\label{equ15}
		\begin{aligned}
        \dot{z}_{i,k}=&(x_{i,k+1}+h_{i,k}(\bar{x}_{i,k})-u_{i,k-1}^*)dt+f_{i,k}(\bar{x}_{i,k})dw_i,\\
        \dot{z}_{i,n}=&(u_i-u_{i,k-1}^*+h_{i,n}(\bar{x}_{i,n})+c_i\xi_{i}(\bar{x}_{i,n},u_i))dt\\&+f_{i,n}(\bar{x}_{i,n})dw_i.
		\end{aligned}
	\end{equation}
    
    Similar to the derivation process in Step 1, here we design the actor-critic-identifier NNs structure including the optimal controller $u_i^*$, the actor NNs $\hat{W}_{u_{i,n}}$, the critic NNs $\hat{W}_{c_{i,n}}$ and the identifier NNs $\hat{W}_{h_{i,n}}$ as follows
    \begin{equation}
		\label{equ16}
		\begin{aligned}
        u_{i}^*&=-\zeta_{i,n}z_{i,n}-\hat{W}^{T}_{h_{i,n}}Q_{h_{i,n}}-\frac{1}{2}\hat{W}_{u_{i,n}}^TQ_{\mathcal{J}_{i,n}},\\
        \dot{\hat{W}}_{u_{i,n}}&=-Q_{\mathcal{J}_{i,n}}^TQ_{\mathcal{J}_{i,n}}(\varphi_{u_{i,n}}(\hat{W}_{u_{i,n}}-\hat{W}_{c_{i,n}})+\varphi_{c_{i,n}}\hat{W}_{c_{i,n}}),\\
        \dot{\hat{W}}_{c_{i,n}}&=-\varphi_{c_{i,n}}Q_{\mathcal{J}_{i,n}}^TQ_{\mathcal{J}_{i,n}}\hat{W}_{c_{i,n}},\\
        \dot{\hat{W}}_{h_{i,n}}&=\delta_{i,n}(Q_{h_{i,n}}z^3_{i,n}-\sigma_{i,n}\hat{W}_{h_{i,n}}),
		\end{aligned}
	\end{equation}
    where $\varphi_{u_{i,n}}$, $\varphi_{c_{i,n}}$, $\delta_{i,n}$ and $\sigma_{i,n}$ are all designed parameters that satisfy $\varphi_{u_{i,n}}>1/2$ and $\varphi_{u_{i,n}}>\varphi_{c_{i,n}}>\varphi_{u_{i,n}}/2$.

    The hybrid ETC strategy is introduced here to reduce the frequency of the controller $u_i^*$, we have
    \begin{equation}
		\label{equ17}
		\begin{aligned}
        u_i(t)=&\kappa_{i}(t_s),\\
        t_{s+1}=& \left\{
        \begin{aligned}
                &\inf \{t \in \mathbb{R}|| \pi_i \mid \geq \lambda_i\},|{u_i}|<\mathcal{H}_g \\
                &\inf \{t \in \mathbb{R}|| \pi_i|\geq \varrho_i|{u_i}| +\lambda_i^*\}, |{u_i}| \geq \mathcal{H}_g
                \end{aligned}
                \right.
		\end{aligned}
	\end{equation}

    When $|{u_i}| < \mathcal{H}_g$, the ETC-based updated controller $\kappa_i$ is designed as
    \begin{equation}
		\label{equ18}
		\begin{aligned}
        \kappa_{i}(t)=&u_i^*-\bar{\lambda}_i\tanh(\frac{\bar{\lambda}_iz_{i,n}}{\Upsilon_i}),\\
		\end{aligned}
	\end{equation}

    When $|{u_i}| \geq \mathcal{H}_g$, controller $\kappa_i$ is designed as 
    \begin{equation}
		\label{equ19}
		\begin{aligned}
        \kappa_{i}(t)=&-(1+\varrho_i)(u_i^*\tanh(\frac{u_i^*z_{i,n}}{\Upsilon_i})+\bar{\varrho}_i\tanh(\frac{\bar{\varrho}_iz_{i,n}}{\Upsilon_i})),\\
		\end{aligned}
	\end{equation}
    where $\varrho_i$, $\bar{\varrho}_i$, $\Upsilon_i$, $\lambda_i$, $\bar{\lambda}_i$ and $\lambda_i^*$ are all designed parameters satisfying $0<\varrho_i<1$ and $\bar{\varrho}_i>\lambda_i^*/(1-\varrho_i)$. And $\pi_i=\kappa_i-u_i$ denotes the error between the ETC-based updated controller $\kappa_i$ and the final controller $u_i$. 

\subsection{Assumptions and Lemmas}
For the development of our results in this paper we make the following assumptions and lemmas.
\begin{ass}[\cite{ass1}]\label{ass1}
      We assume that $y_r(t)$ and $\dot{y}_r(t)$ are bounded and known.
\end{ass}

\begin{ass}[\cite{ass2}]\label{ass2}
      We assume that existing an unknown nonnegative function $\Phi_i(\bar{x}_{i,n},u_i)$, satisfying the subsequent inequality condition
                \begin{equation}
            		\label{ass}
            		\begin{aligned}
                    |h_{i,n}(\hat{x}_{i,n})+c_i\xi_{i}(\bar{x}_{i,n},u_i)|\leq \Phi_i(\bar{x}_{i,n},u_i).
            		\end{aligned}
            	\end{equation}
\end{ass}

Assumptions~\ref{ass1} and~\ref{ass2} are employed to guarantee that the desired trajectory and the uncertainties are strictly bounded. Note that Assumptions~\ref{ass1} and~\ref{ass2} are standard for consensus tracking control of plants with non-specific fault-induced uncertainties, and similar assumptions can be found in\cite{ass1,ass2,ass3}. Without these assumptions, the proposed optimal distributed control algorithm cannot be realized.

\begin{lemma}[\cite{HETC2}]\label{lemma111}
        For any constant $\chi_1 \in \mathbb{R}$ and any positive constant $\chi_2 > 0$, the inequality $0\leq |\chi_1|-\chi_1\tanh\left(\frac{\chi_1}{\chi_2}\right)\leq\varepsilon^*\chi_2$ holds, where $\varepsilon^*$ is a constant satisfying $\varepsilon^*=e^{-(\varepsilon^*+1)}$ (approximately $\varepsilon^* \approx 0.2785$).
\end{lemma}

\begin{lemma}[\cite{intro8,HETC2}]\label{lemma222}
        For any Hurwitz matrix $ A \in \mathbb{R} $ and any symmetric positive definite matrix $ B \in \mathbb{R} $, there exists a unique symmetric positive definite matrix $ A^* \in \mathbb{R} $ satisfying $A^T A^* + A^* A = -B$.
\end{lemma}

\begin{lemma}[\cite{lemma1,lemma2}]\label{lemma333}
        Suppose there exists a $C^2$ function $V(\chi)$ satisfying
        \begin{equation}
            		\label{lemma1}
            		\begin{aligned}
                    \Xi_1(\chi)\leq V(\chi)\leq\Xi_2(\chi),
            		\end{aligned}
        \end{equation}
        where $\Xi_1(\chi)$ and $\Xi_2(\chi)$ are class $K_{\infty}$ functions. The infinitesimal generator of $V(\chi)$ satisfying
        \begin{equation}
            		\label{lemma2}
            		\begin{aligned}
                    \dot{V}(\chi)\leq-\Psi_1V(\chi)+\Psi_2,
            		\end{aligned}
        \end{equation}
        where $\Psi_1$ and $\Psi_2$ are two positive constants. Then for $\forall \chi,\chi_0\in\mathbb{R}$, there is a unique strong solution such that
        \begin{equation}
            		\label{lemma3}
            		\begin{aligned}
                    {V}(\chi)\leq e^{-\Psi_1t}V(0)+\frac{\Psi_2}{\Psi_1}.
            		\end{aligned}
        \end{equation}

        If inequality~\eqref{lemma3} holds, the states in $V(\chi)$ are all SGUUB in mean square.
\end{lemma}

Lemma~\ref{lemma111} is employed in the stability analysis for the substitution and scaling between the ETC-based updated controller $\kappa_i$ and the final controller $u_i$.
Lemma~\ref{lemma222} is used in the stability analysis for matrix elimination with respect to the NN weights.
Lemma~\ref{lemma333} is applied in the Lyapunov-based stability analysis.

\subsection{Stability Analysis}
In this part, we use Lyapunov-based stability analysis to establish theoretical results, with Theorem \ref{theorem} and proof introduced subsequently.
\begin{theorem}\label{theorem}
       Consider the MASs~\eqref{equ1} with stochastic dynamics $w_i\in\mathbb{R}$ and the non-specific fault-induced uncertainties $\xi_i\in\mathbb{R}$, the designed optimal controller~\eqref{equ16}, the actor-critic-identifier NNs structure~\eqref{equ16} and the hybrid ETC strategy~\eqref{equ17}~\eqref{equ18}~\eqref{equ19}. During the control process, all errors within MASs~\eqref{equ1} are SGUUB. In other words, all followers can track the leader’s desired trajectory. Meanwhile, there exists a constant $t^{*}>0$ such that all inter-execution intervals $\{t_{s+1}-t_{s}\}$ are uniformly lower bounded by $t^{*}$.
\end{theorem}

\textit{Proof:}
        To ensure the convergence of all signals in the closed-loop MASs~\eqref{equ1}, we construct the following Lyapunov function candidate as follows
        \begin{equation}
            		\label{equ24}
            		\begin{aligned}
                    V=&\sum\limits_{i=1}^{N}\sum\limits_{k=1}^{n} V_{i,k},\\
                    V_{i,k}=&\frac{1}{4}z_{i,k}^4+\frac{1}{2}\tilde{W}^T_{u_{i,k}}\tilde{W}_{u_{i,k}}+\frac{1}{2}\tilde{W}^T_{c_{i,k}}\tilde{W}_{c_{i,k}}\\&+\frac{1}{2}\tilde{W}^T_{h_{i,k}}\delta_{i,k}^{-1}\tilde{W}_{h_{i,k}},
            		\end{aligned}
        \end{equation}
        where $k=1,...,n$. According to the hybrid ETC strategy~\eqref{equ19}, when $|u_i|\geq \mathcal{H}_g$, we define two time-varying constants $|\varpi_{i,1}(t)|\leq1$ and $|\varpi_{i,2}(t)|\leq1$, satisfying $\varpi_{i,1}(t_s)=\varpi_{i,2}(t_s)=0$, $\varpi_{i,1}(t_{s+1})=\varpi_{i,2}(t_{s+1})=\pm1$, and we have $\kappa_i(t)=(1+\varpi_{i,1}\varrho_i)u^*_i(t)+\varpi_{i,2}\lambda_i^*$. Based on Lemma~\ref{lemma111}, we have
        \begin{equation}
		\label{equ25}
		\begin{aligned}
			z_{i,k}\frac{\kappa_i(t)-\varpi_{i,2}\lambda_i^*}{1+\varpi_{i,1}\varrho_i}\leq 0.557\Upsilon_i. 
		\end{aligned}
	    \end{equation}

        When $|u_i|< \mathcal{H}_g$, we define a time-varying constant $\vartheta_i(t)\leq1$, satisfying $\vartheta_i(t_s)=0$ and $\vartheta_i(t_{s+1})=\pm1$, then based on~\eqref{equ18} and Lemma~\ref{lemma111}, we have
        \begin{equation}
		\label{equ26}
		\begin{aligned}
			z_{i,k}(\bar{\lambda}_i\tanh(\frac{\bar{\lambda}_iz_{i,n}}{\Upsilon_i})-\vartheta_i(t)\lambda_i)\leq 0.2785\Upsilon_i.
		\end{aligned}
	    \end{equation}

        Furthermore, based on the Young’s inequality and Assumption~\ref{ass2}, we have 
        \begin{equation}
		\label{equ27}
		\begin{aligned}
			z_{i,k}^3\rho_{h_{i,k}}\leq&\frac{1}{2}z_{i,k}^6+\frac{1}{2}\rho^2_{h_{i,k}}\\
            \eta_i z_{i,k}^3z_{i,k+1}\leq&\frac{3}{4}\eta_i z_{i,k}^4+\frac{1}{4}\eta_i z_{i,k+1}^4\\
            z_{i,k}^3(h_{i,n}+c_i\xi_{i})\leq&z_{i,k}^3\Phi_i+\frac{1}{2}\\
            -\frac{1}{2}z_{i,k}^3\hat{W}^T_{u_{i,k}}Q_{\mathcal{J}_{i,k}}\leq&\frac{1}{4}z_{i,k}^6+\frac{1}{4}(\hat{W}^T_{u_{i,k}}Q_{\mathcal{J}_{i,k}})^2
		\end{aligned}
	    \end{equation}
        
        Then considering Assumption~\ref{ass1}, Lemma~\ref{lemma222} and utilizing~\eqref{equ15}-\eqref{equ19},~\eqref{equ25}-\eqref{equ27}, we take the derivative of Lyapunov function $V$ as
        \begin{equation}
		\label{equ28}
		\begin{aligned}
			\dot{V}\leq&-\sum\limits_{i=1}^{N}\sum\limits_{k=1}^{n}\zeta_{i,k}z_{i,k}^4-\sum\limits_{i=1}^{N}\sum\limits_{k=1}^{n}\frac{\varphi_{u_{i,k}}\Theta_{min}(\mathcal{J})}{2}\tilde{W}^T_{u_{i,k}}\tilde{W}_{u_{i,k}}\\
            &-\sum\limits_{i=1}^{N}\sum\limits_{k=1}^{n}\frac{\varphi_{c_{i,k}}\Theta_{min}(\mathcal{J})}{2}\tilde{W}^T_{c_{i,k}}\tilde{W}_{c_{i,k}}+\sum\limits_{i=1}^{N}\sum\limits_{k=1}^{n}\frac{1}{2}\rho^2_{h_{i,k}}\\
            &-\sum\limits_{i=1}^{N}\sum\limits_{k=1}^{n}\frac{\sigma_{i,k}}{2\Theta_{max}(\delta)}\tilde{W}^T_{h_{i,k}}\delta_{i,k}^{-1}\tilde{W}_{h_{i,k}}+\sum\limits_{i=1}^{N}\sum\limits_{k=1}^{n}\frac{1}{2}z_{i,k}^6\\
            &+\sum\limits_{i=1}^{N}\sum\limits_{k=1}^{n}\frac{\sigma_{i,k}}{2}{W}^{*T}_{h_{i,k}}{W}^*_{h_{i,k}}+\sum\limits_{i=1}^{N}\sum\limits_{k=1}^{n}\frac{1}{4}(\hat{W}^T_{u_{i,k}}Q_{\mathcal{J}_{i,k}})^2\\
            &+\sum\limits_{i=1}^{N}\sum\limits_{k=1}^{n}\frac{\varphi_{u_{i,k}}+\varphi_{c_{i,k}}}{2}||W^{*T}_{\mathcal{J}_{i,1}}Q_{\mathcal{J}_{i,1}}||^2+\sum\limits_{i=1}^{N}\frac{1}{4}\eta_i z_{i,2}^4\\
            &+\sum\limits_{i=1}^{N}\sum\limits_{k=2}^{n}z_{i,k+1}^4+\sum\limits_{i=1}^{N}0.8355\Upsilon_i,
		\end{aligned}
	    \end{equation}
        where the maximum and minimum eigenvalues of $\delta_{i,k}^{-1}$ and $W_{\mathcal{J}_{i,k}}W^{T}_{\mathcal{J}_{i,k}}$ are defined as $\Theta_{max}(\delta)$ and $\Theta_{min}(\mathcal{J})$. Then we have
        \begin{equation}
		\label{equ29}
		\begin{aligned}
			\dot{V}\leq -\wp V+\Delta
		\end{aligned}
	    \end{equation}
        where
        \begin{equation}
		\label{equ30}
		\begin{aligned}
			\wp=&\min\{4\zeta_{i,k},\varphi_{u_{i,k}}\Theta_{min}(\mathcal{J}),\varphi_{c_{i,k}}\Theta_{min}(\mathcal{J}),\frac{\sigma_{i,k}}{\Theta_{max}(\delta)}\},\\
            \Delta=&\sum\limits_{i=1}^{N}\sum\limits_{k=1}^{n}[\frac{\rho^2_{h_{i,k}}+z_{i,k}^6}{2}+\frac{(\hat{W}^T_{u_{i,k}}Q_{\mathcal{J}_{i,k}})^2+\eta_i z_{i,2}^4}{4}+z_{i,k+1}^4\\
            &+\frac{\varphi_{u_{i,k}}+\varphi_{c_{i,k}}}{2}||W^{*T}_{\mathcal{J}_{i,1}}Q_{\mathcal{J}_{i,1}}||^2]+\sum\limits_{i=1}^{N}0.8355\Upsilon_i.
		\end{aligned}
	    \end{equation}
              
        And \eqref{equ30} satisfies
        \begin{equation}
		\label{equ31}
		\begin{aligned}
			0\leq V(t)\leq e^{-\wp t}V(0)+\frac{\Delta}{\wp}.
		\end{aligned}
	    \end{equation}
        
        Based on \eqref{equ31} and Lemma~\ref{lemma333}, the tracking error $z_{i,k}$ and NNs approximation errors $\tilde{W}^T_{u_{i,k}}$, $\tilde{W}_{c_{i,k}}$, $\tilde{W}_{h_{i,k}}$ are SGUUB. 
        Since all closed-loop errors are bounded and based on \eqref{equ18}\eqref{equ19}, we have that $\kappa_i(t)$ is continuous and bounded. Furthermore, observing that $\pi_i(t_s)=0$ and $\lim_{t\rightarrow t_{s+1}}\pi_i(t)=\max\{\lambda_i,\varrho_i|{u_i}| +\lambda_i^*\}$, there exists a constant $\Lambda_i \in \mathbb{R}^+$ ensuring that the event-triggered minimum inter-execution interval $t^*$ satisfies $t^*\geq\max\{\lambda_i,\varrho_i|{u_i}| +\lambda_i^*\}/\Lambda_i$. This means that the Zeno behavior is successfully avoided.
        $\hfill\square$

        In Theorem~\ref{theorem}, we established that the proposed RL-based optimal distributed control algorithm enables the followers in MASs~\eqref{equ1} to achieve consensus tracking while keeping tracking errors bounded.
        Specifically, the proof introduces the Lyapunov function in \eqref{equ24}, which is constructed based on the tracking errors and NNs approximation errors. 
        Then, by applying Lyapunov stability analysis, we proved through \eqref{equ29} and \eqref{equ31} that all tracking errors are SGUUB. 
        Compared with the previous work in~\cite{compare}, the present analysis differs in \eqref{equ16}, because the control law incorporates an RL-based actor–critic–identifier neural network structure.
        In addition, unlike the static triggering threshold or dynamic triggering threshold adopted in~\cite{intro7,intro8,ass1}, the threshold in our proposed hybrid event-triggering condition varies with $\mathcal{H}_g$. Moreover, the adopted hybrid ETC strategy is shown to exclude Zeno behavior.

\section{Simulation Results}\label{sec4}
In this section, we present the numerical simulation in terms of the single-axis robotic manipulator. It is organized as follows: we first introduce the parameter setting for the control algorithm. and we analyze the simulation results, including the tracking performance (Fig.~\ref{tracking}), the NNs weight norms (Fig.~\ref{NNs}), the hybrid ETC statistics and update intervals (Table~\ref{table}, Fig.~\ref{controller}, Fig.~\ref{etc}). Then we propose a quantitative analysis method rooted in
economics for ETC (Fig.~\ref{ku}). Finally, we discuss the comparisons (Fig.~\ref{comp}, Fig.~\ref{errorcomp} and Table~\ref{table}).
	
\textbf{Parameter Setting.} We now demonstrate the practical applicability of the single-axis robotic manipulator \cite{simu1,simu2} with stochastic uncertainties as
    \begin{equation}
		\label{equ_simu}
		\begin{aligned}
			\mathcal{I} \ddot{\theta}_i(t) + \mathcal{Q} \dot{\theta}_i(t) &+ \mathcal{M}g\mathcal{K}\sin(\theta_i(t)) \\ &+ f(\theta_i,\dot{\theta}_i,u)+c_i \xi(\theta_i,\dot{\theta}_i,u) - u(t) =0
		\end{aligned}
	    \end{equation}
    where $\theta_i(t)$, \(\dot{\theta}_i(t)\), \(\ddot{\theta}_i(t)\) denote the robotic manipulator joint angle, angle velocity and angle acceleration, respectively. In this example, $\mathcal{I}=10.5~kg\cdot m^2$ represents the moment of inertia, $\mathcal{Q}=30.5~N\cdot m\cdot s/rad$ represents the viscous friction, \(\mathcal{M}=1~kg\) is the mass, \(\mathcal{K}=1.5~m\) is the length, $g=9.8~m/s^2$ is the gravitational acceleration, $f(\theta_i,\dot{\theta}_i,u)=\dot{\theta}_i^2\sin(\theta_i)+0.1\sin(t)\cos(\theta_i\dot{\theta}_i)$ denotes the unknown stochastic disturbances, $c_i=1-\mathrm{e}^{-\mu\left(t-T_0\right)}$ with $\mu=15$ and $T_0=10~s$, $\xi(\theta_i,\dot{\theta}_i,u)=10(\theta_i\dot{\theta}_i+\sin(u))+8$ denotes the non-specific fault-induced uncertainties and \(u(t)\) is the control input used to represent the motor torque. The control objective is to make the joint angle $\theta_i$ $(i=1,2,3.)$ track the desired trajectory $\theta_r=0.8\cos(2t)\sin(t+0.5t^2)+0.2\sin^2(t)+0.1$.

    Based on the communication topology in Fig.~\ref{topo}, we design the matrices of communication network as: the adjacency matrix $\mathcal{A}=[0,0,0;1,0,0;1,0,0]$, the Laplacian matrix $\mathcal{L}=[0,0,0;-1,1,0;-1,0,1]$ and the connection matrix $\mathcal{B}=\operatorname{diag}(1,0,0)$.

    The initial position for each follower robotic manipulator are $\theta_1(0)=0$, $\theta_2(0)=-0.1$ and $\theta_3(0)=0.2$. Furthermore, the designed parameters utilized in the control process are outlined below: the tracking error feedback gain $\zeta_{i,k}=1.2$; the actor-critic-identifier NNs adaptation/leakage gains are $\varphi_{u_{i,k}}=13$, $\varphi_{c_{i,k}}=15$, $\delta_{i,k}=\sigma_{i,k}=1.5$; the hybrid ETC strategy constants are $\mathcal{H}_g=6$, $\varrho_i=0.28$, $\bar{\varrho}_i=3$, $\Upsilon_i=0.3$, $\lambda_i=5.2$, $\bar{\lambda}_i=3$ and $\lambda_i^*=3.2$. In addition, the duration of the simulation experiment is $15~\mathrm{s}$, with a time-step of 0.01. 
    
        \begin{figure}
            \centering
            \includegraphics[width=0.45\textwidth, height=0.11\textheight]{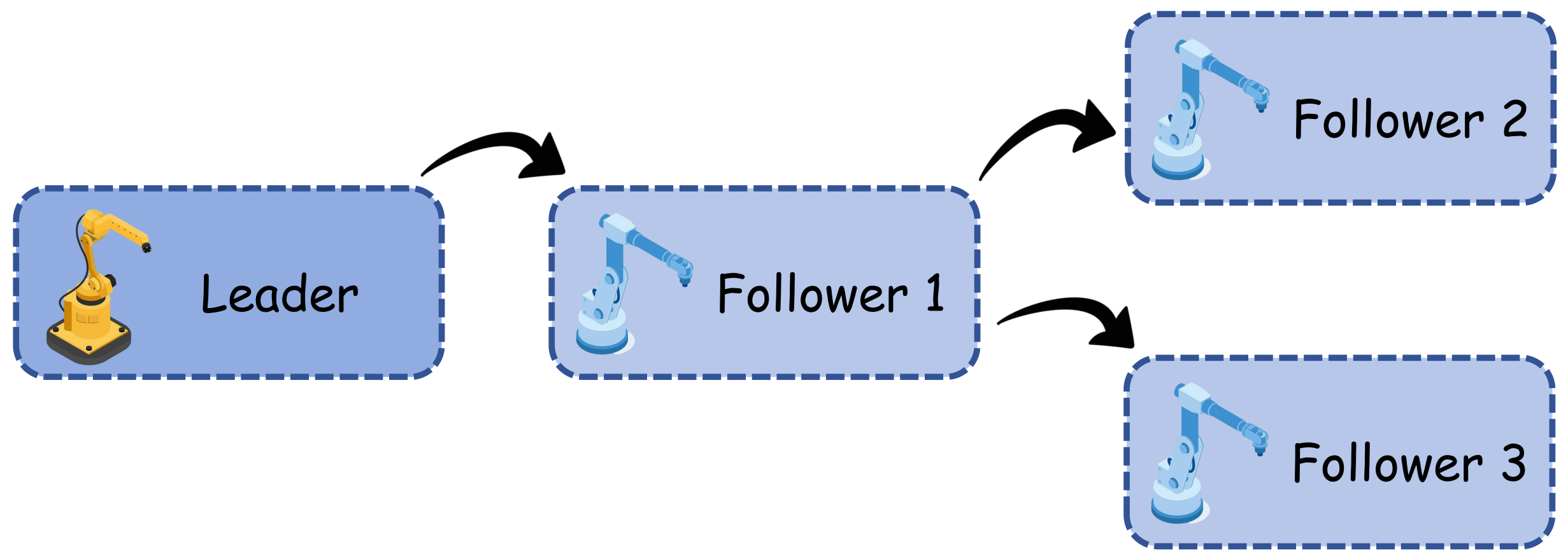}
            \caption{Communication topology}
            \label{topo}
        \end{figure}

        \begin{figure}
            \centering
            \includegraphics[width=0.48\textwidth, height=0.25\textheight]{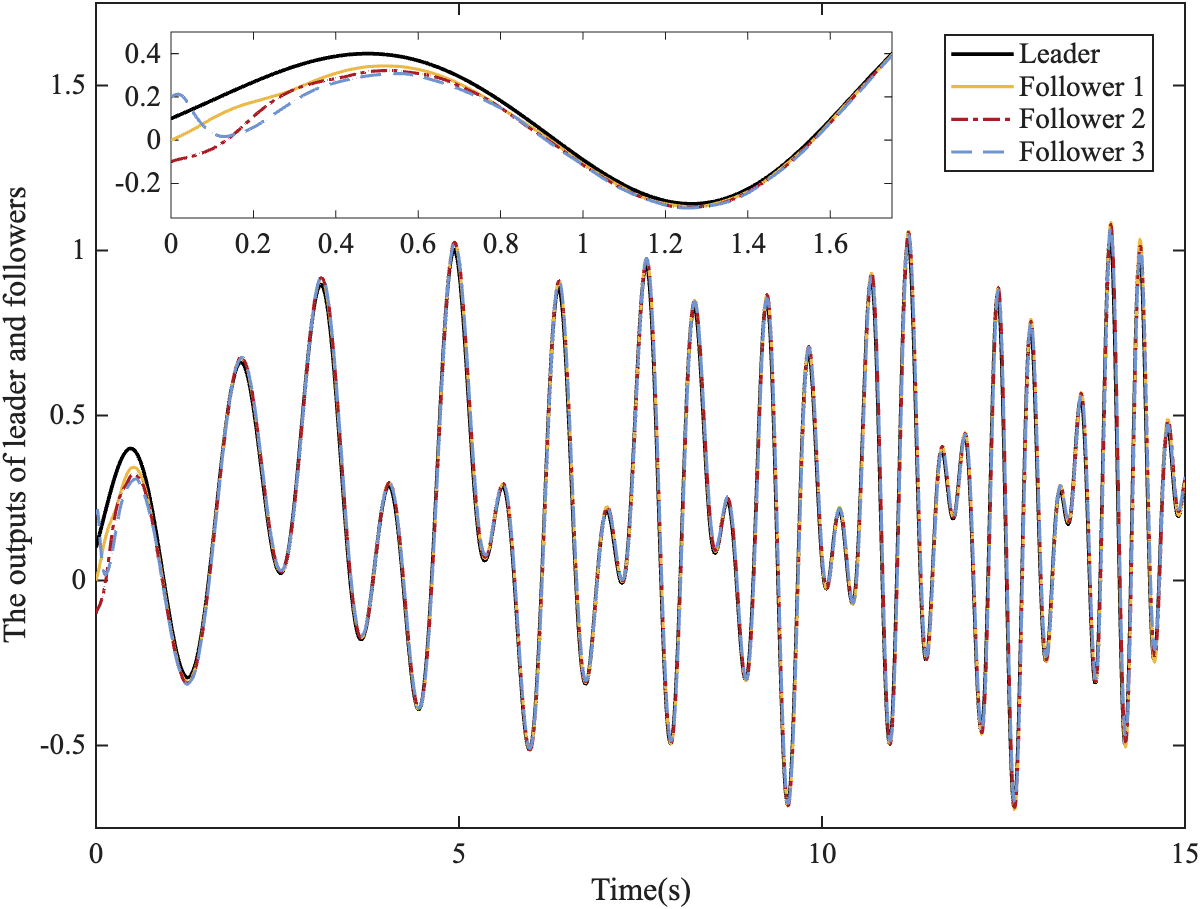}
            \caption{The tracking performance under the RL-based optimal distributed ETC algorithm.}
            \label{tracking}
        \end{figure}


        \begin{figure}
            \centering
            \includegraphics[width=0.48\textwidth, height=0.27\textheight]{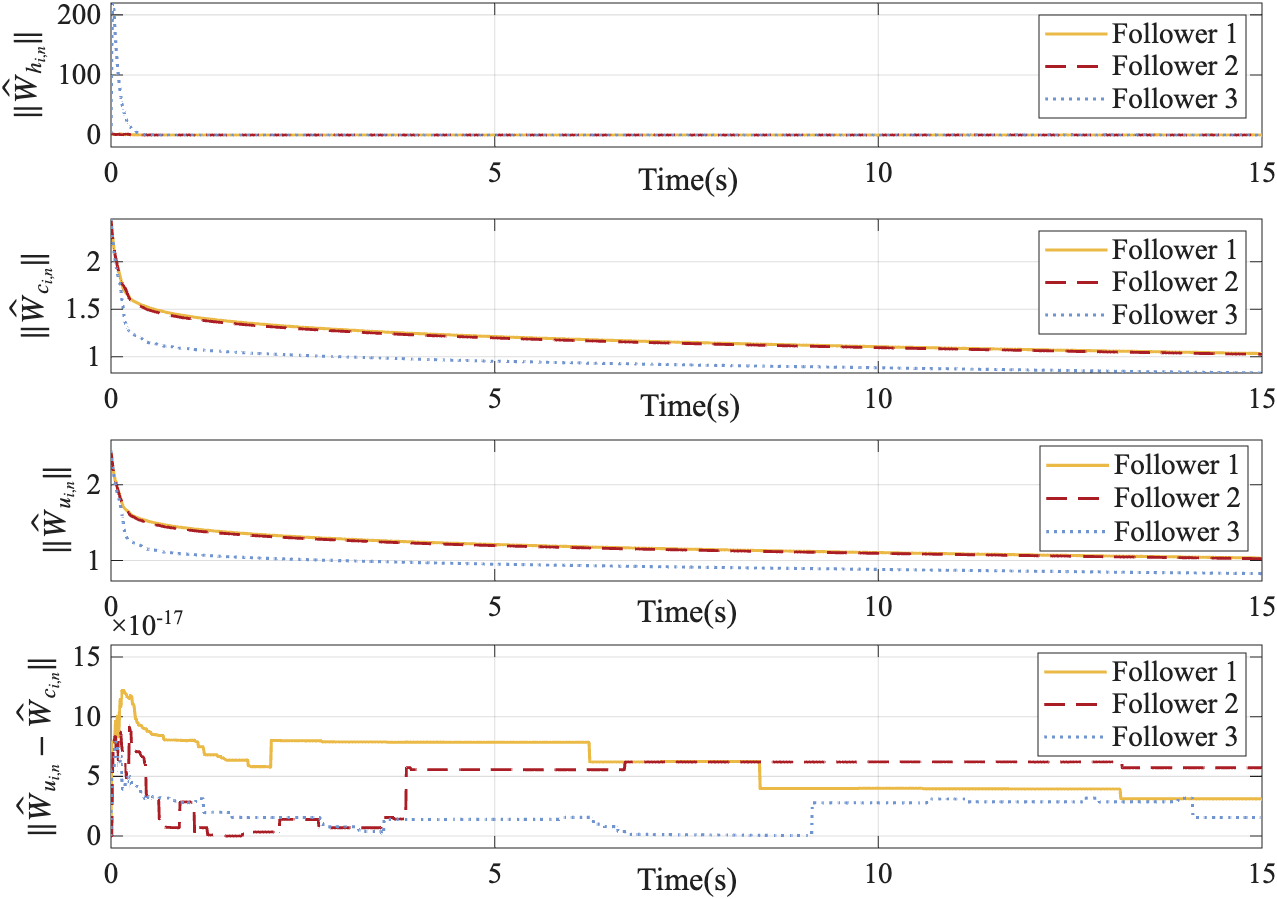}
            \caption{The weight norms of the actor, critic, identifier NNs and the actor-critic weight gap.}
            \label{NNs}
        \end{figure}

        \begin{figure}
            \centering
            \includegraphics[width=0.48\textwidth, height=0.25\textheight]{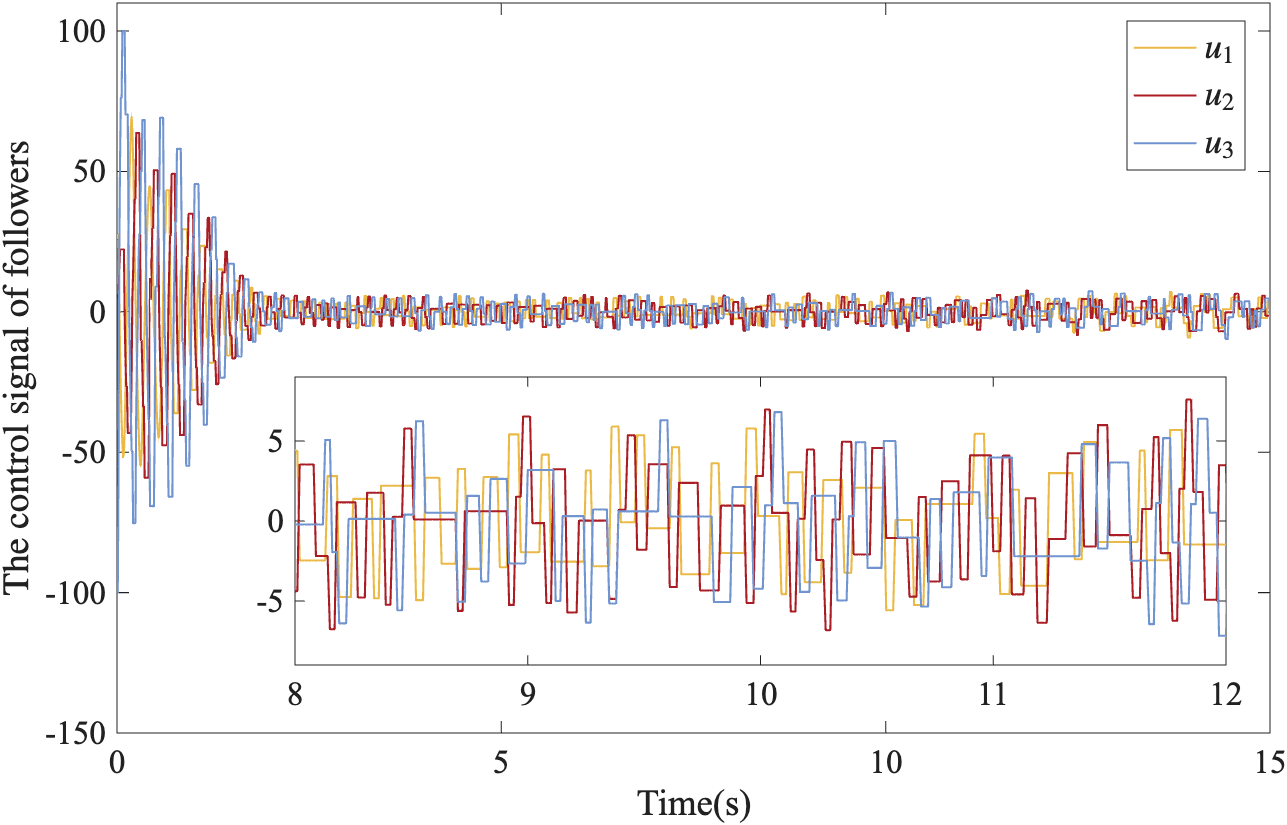}
            \caption{The control signal of followers under the hybrid ETC strategy.}
            \label{controller}
        \end{figure}

        \begin{figure}
            \centering
            \includegraphics[width=0.48\textwidth, height=0.25\textheight]{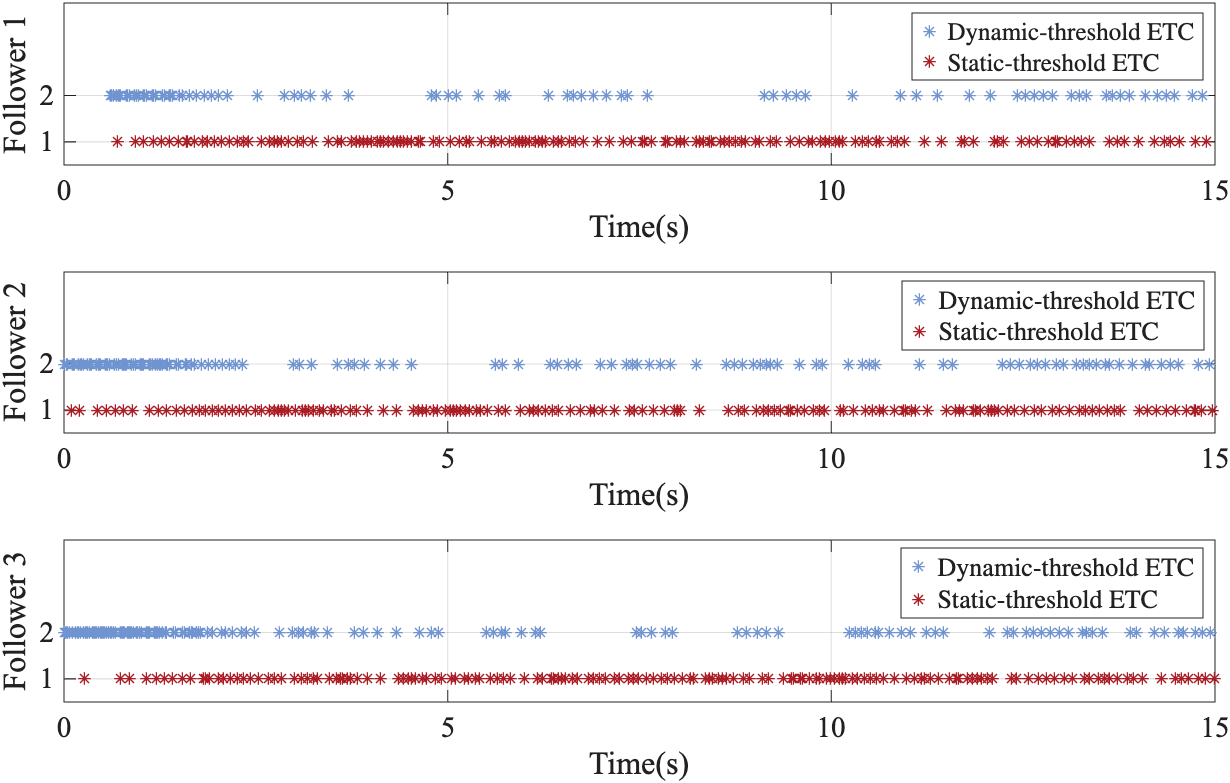}
            \caption{The intervals of the hybrid ETC strategy for followers.}
            \label{etc}
        \end{figure}

        \textbf{Results Analysis.} We analyze the simulation results in conjunction with the Figs.~\ref{tracking} -~\ref{etc} and Table \ref{table}. Fig.~\ref{tracking} illustrates the tracking performance of the leader's trajectory path $y_r$ ($\theta_r$ in simulation) in relation to three followers under the single-axis robotic manipulator MASs optimal distributed control algorithm. With the controller in effect, the tracking task is accomplished within 2 seconds. 
        Fig.~\ref{NNs} presents the evolution of actor, critic, and identifier weight norms. The identifier weights $\hat{W}_{h_{i,n}}$ rapidly adjust and then remain bounded, indicating stable approximation of unknown dynamics and uncertainties. The critic weights $\hat{W}_{c_{i,n}}$ and actor weights $\hat{W}_{u_{i,n}}$ decrease smoothly and stay bounded, confirming stable performance evaluation and policy learning. Moreover, the actor-critic gap remains nearly zero, showing strong consistency between the learned policy and critic evaluation. These results support boundedness of NNs and verify effectiveness of the actor-critic-identifier structure.
        The control inputs of three followers are shown in Fig.~\ref{controller}. It can be observed that, during the control process, the control signals become intermittent under the effect of ETC strategy.
        Fig.~\ref{etc} presents the triggering intervals of the hybrid ETC strategy for each follower under both the static-threshold and dynamic-threshold ETC cases, where the absence of fully overlapping triggering instants provides simulation-based evidence for the avoidance of Zeno behavior.

        In Table \ref{table}, the rate is calculated as
    \begin{equation}
		\nonumber
		\begin{aligned}
			\text{Rate}=\frac{\text{Continuous counts-Hybrid counts}}{\text{Continuous counts}}\times 100\%,
		\end{aligned}
	\end{equation}
    which represents the proportion of control resources saved by the hybrid ETC strategy.

    \textbf{Quantitative Analysis of ETC Performance.} The above ETC based analysis is consistent with that in \cite{simu2,simu3}. However, if only the trigger count is considered, the most ‘cost-effective’ strategy would be to trigger almost never, which may significantly degrade tracking performance. Therefore, this paper introduces a comprehensive index $K_u$ rooted in economics~\cite{econ1,econ2}, incorporating triggering costs $K_{\text{trigger}}$, performance degradation costs $K_{\text{track}}$, and control actuation wear $K_{\text{control}}$. We have

    \begin{equation}
		\nonumber
		\begin{aligned}
			K_u=\alpha K_{\text{trigger}}+\beta K_{\text{track}}+\gamma K_{\text{control}},
		\end{aligned}
	\end{equation}
    where $0<\alpha, \beta, \gamma<1$ denote the designed weight parameters, satisfying $\alpha+\beta+\gamma=1$. Here, triggering costs $K_{\text{trigger}}$ reflect the degree of control-resource savings achieved by the proposed strategy. Performance degradation costs $K_{\text{track}}$ indicate whether the tracking task deteriorates under the influence of ETC. Control actuation wear $K_{\text{control}}$ represents the higher energy consumption, heat generation, and sustained actuator load that may arise from the output of control signals. Then, we introduce the economic concept of benchmark normalization~\cite{OR} and convert the above three metrics into ratios of ETC relative to continuous control. Accordingly, we define $K_{\text{trigger}}$, $K_{\text{track}}$, and $K_{\text{control}}$ as follows:
        \begin{equation}
    		\nonumber
    		\begin{aligned}
    			K_{\text{trigger}}=\frac{\text{Continuous counts}-\text{Hybrid counts}}{\text{Continuous counts}},
    		\end{aligned}
    	\end{equation}
        \begin{equation}
    		\nonumber
    		\begin{aligned}
    			K_{\text{track}}=\frac{\int_{0}^{T}e_i^2(t)dt|_{\text{ETC}}}{\int_{0}^{T}e_i^2(t)dt|_{\text{Continuous control}}},
    		\end{aligned}
    	\end{equation}
        \begin{equation}
    		\nonumber
    		\begin{aligned}
    			K_{\text{control}}=\frac{\int_{0}^{T}u_i^2(t)dt|_{\text{ETC}}}{\int_{0}^{T}u_i^2(t)dt|_{\text{Continuous control}}}.
    		\end{aligned}
    	\end{equation}

        Beyond the triggering-count statistics, the benchmark-normalized index $K_u$ provides a more informative evaluation of the hybrid ETC strategy from the perspectives of communication efficiency, tracking accuracy, and control effort. 
        First, we perform theoretical analysis: The value of $K_{\text{trigger}}$ is expected to converge to a stable level close to $1$, which indicates that the number of triggered controller signal updates is relatively small and that a higher degree of control resource savings is achieved. The value of $K_{\text{track}}$ should remain around $1$ after the transient stage, suggesting that the discrete implementation introduced by the ETC strategy has only a marginal influence on the tracking task compared with the continuous control benchmark. And the value of $K_{\text{control}}$ should also stay close to or below $1$, indicating that the ETC strategy cannot introduce additional actuator burden or induce excessive control oscillations.
        Therefore, a stable $K_u$ trajectory within a narrow neighborhood of $1$ indicates that the ETC strategy can substantially reduce controller update frequency while preserving tracking accuracy and control effort at levels comparable to the continuous-control benchmark.

        In this paper, the weight parameters for the proposed analytical framework are set as $\alpha = 0.3$, $\beta = 0.6$, and $\gamma = 0.1$. This setting emphasizes performance preservation while still accounting for communication efficiency and energy usage.

        We present the time-varying $K_u$ values of all followers and their average response during the control process in Fig.~\ref{ku}. All followers exhibit a rapid transient adjustment and then converge to a narrow neighborhood of the benchmark level of value $1$. After about 5 seconds, the curves remain within the shaded $0.9$ to $1.1$ band, indicating that the hybrid ETC strategy maintains a stable trade off among triggering reduction, tracking accuracy, and control effort. The average curve remains near $1.05$, confirming overall balanced performance without long term degradation.

        \begin{figure}
            \centering
            \includegraphics[width=0.48\textwidth, height=0.25\textheight]{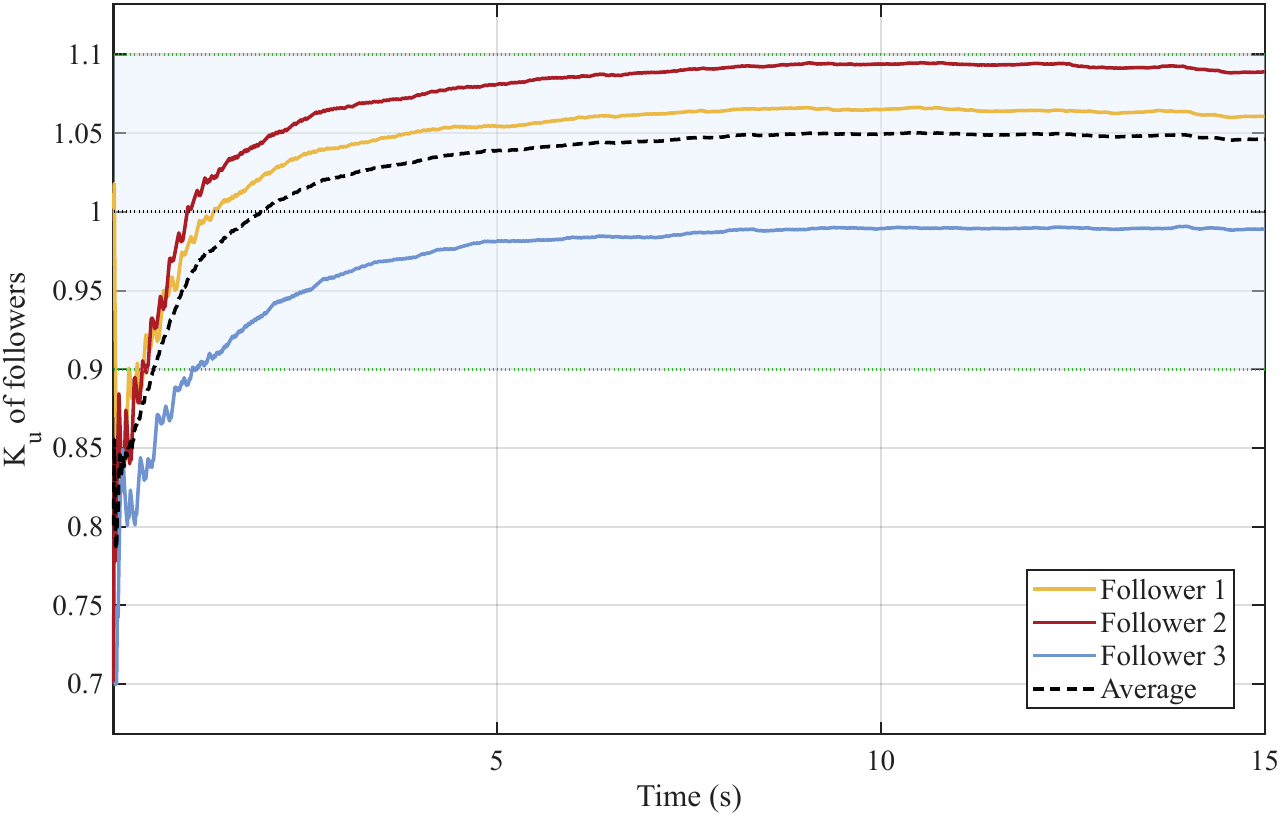}
            \caption{The value of $K_u$ for each follower in the control process.}
            \label{ku}
        \end{figure}

        \textbf{Comparisons.} In Fig.~\ref{comp}, we compare the tracking performance of the proposed optimal controller with a non-optimal baseline (in~\cite{compare}) under the same initial conditions, stochastic disturbances, fault-induced uncertainties, and hybrid event-triggered mechanism. Both controllers can drive the followers to track the leader trajectory, which verifies the basic effectiveness of the distributed control framework. However, the proposed RL-based optimal controller exhibits a visibly faster transient response and smaller tracking deviations, especially during the initial stage. The tracking error comparison in Fig.~\ref{errorcomp} further reveals that the proposed optimal controller achieves improved transient regulation for all followers. In particular, the proposed optimal control curves exhibit smaller initial deviations, faster attenuation of tracking errors, and a more compact bounded-error region throughout the simulation. This improvement benefits from the actor-critic structure, where the critic evaluates the long-term cost and the actor updates the control policy accordingly, leading to more performance-oriented control actions under stochastic disturbances and fault-induced uncertainties.

        \begin{table*}
		\caption{Triggering counts across the hybrid ETC strategy.}
		\label{table}
		\centering
		\footnotesize
		\begin{tabular}{cccccccc}
			\toprule
			\textbf{No.} & \textbf{Static} & \textbf{Dynamic} & \textbf{Hybrid} & \textbf{Continuous} & \textbf{Rate} & Rate in~\cite{intro8} & Rate in~\cite{simu2}\\
			\midrule
			Follower~1 & 88 & 114 & 202 & 1500 & \textbf{86.53\%$\uparrow$} & 85.93\% & 82.71\%\\			
            Follower~2 & 134 & 148 & 282 & 1500 & \textbf{81.20\%$\uparrow$} & 51.37\% & 79.27\%\\
			Follower~3 & 150 & 153 & 303 & 1500 & \textbf{79.80\%$\uparrow$} & 42.35\% & 71.36\%\\
			\bottomrule
		\end{tabular}
	\end{table*}

        In Table \ref{table}, compare with~\cite{intro8,simu2}, the proposed hybrid ETC strategy achieves a higher reduction rate in controller activation frequency. In addition, compare our proposed benchmark-normalized quantitative analysis method about ETC strategy with the existing methods in \cite{intro8,HETC2,simu2,simu3}, this analytical framework overcomes the limitation of conventional ETC strategy that focus only on the reduction of triggering frequency. It quantitatively integrates communication efficiency, tracking accuracy, and control effort into a unified assessment criterion. Moreover, by adjusting the weighting parameters ($\alpha, \beta, \gamma$ in this paper), the relative importance of different control objectives can be flexibly regulated according to specific control requirements. Furthermore, the proposed ETC analysis method offers a quantitative comparison perspective for future studies on ETC strategy.

        \begin{figure}
            \centering
            \subfloat[Tracking under RL-based optimal control algorithm.]{\includegraphics[width=0.975\linewidth]{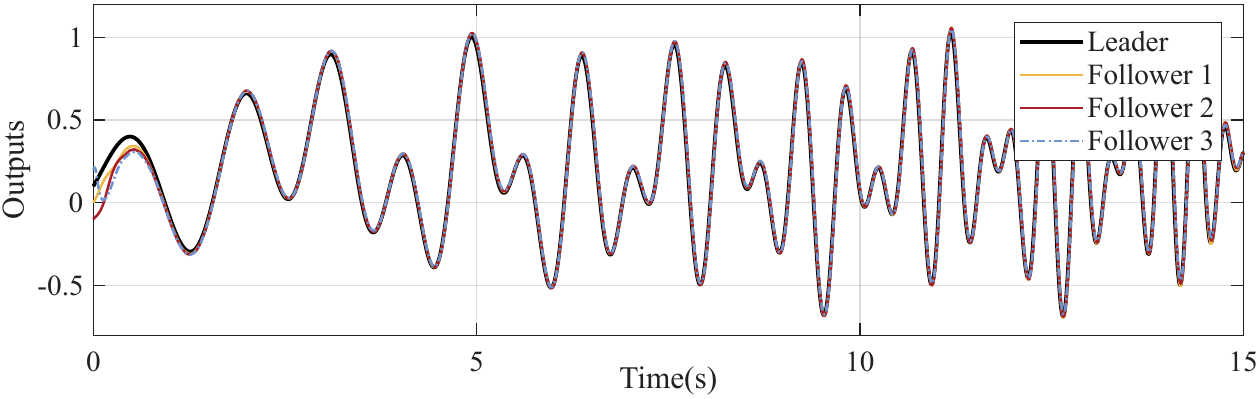}}\\
            \subfloat[Tracking under non-optimal control algorithm.]{\includegraphics[width=0.975\linewidth]{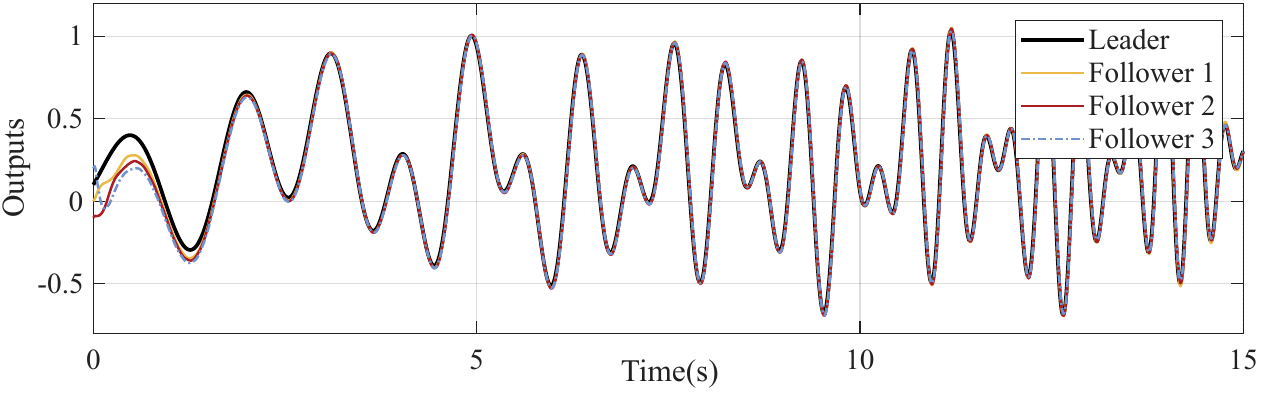}}
            \caption{Comparison of tracking performance between the proposed RL-based optimal control algorithm and non-optimal control algorithm in~\cite{compare}.}
            \label{comp}
        \end{figure}

        \begin{figure}
            \centering
            \includegraphics[width=0.48\textwidth, height=0.3\textheight]{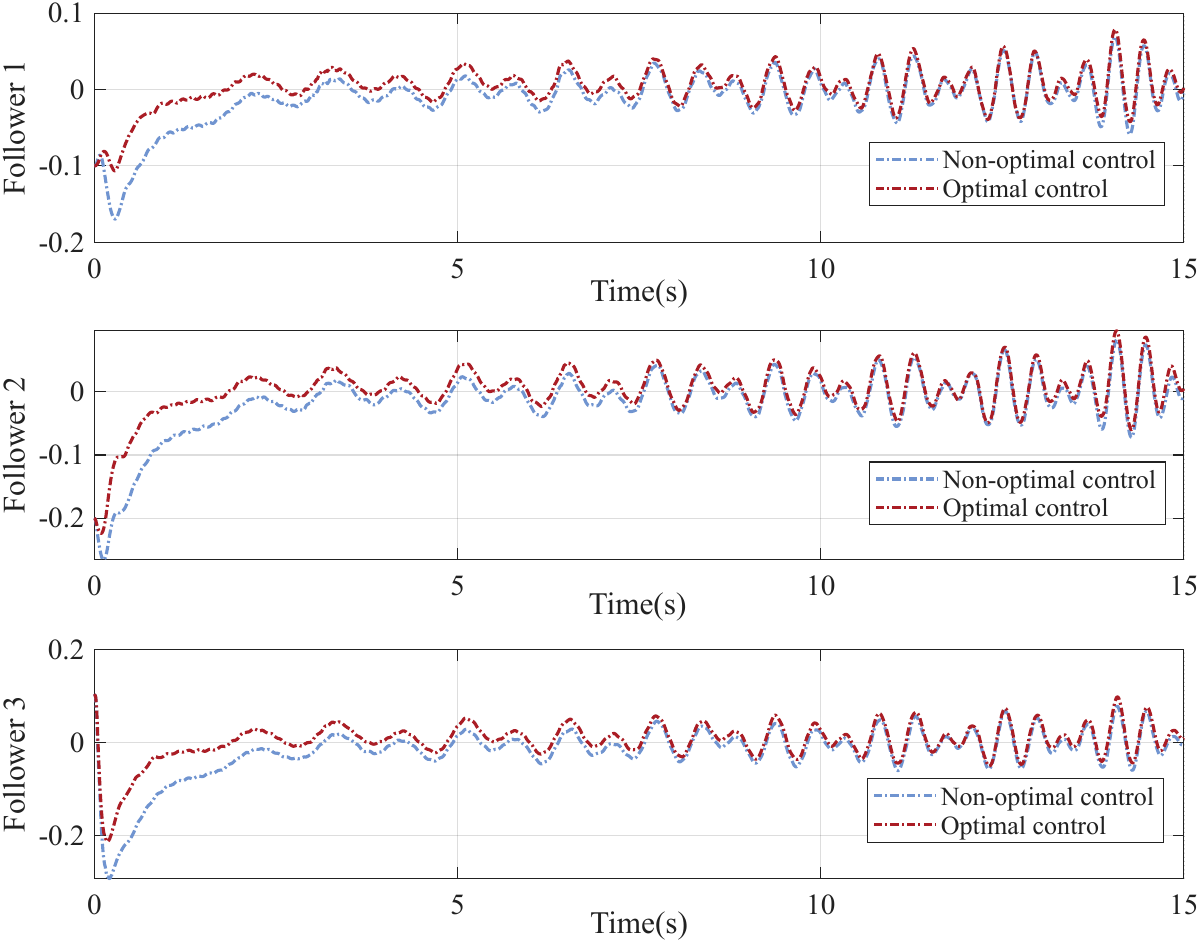}
            \caption{Comparison of tracking errors for followers between the proposed RL-based optimal control algorithm and non-optimal control algorithm in~\cite{compare}.}
            \label{errorcomp}
        \end{figure}
       
\section{Conclusion and Future Directions}\label{sec5}
This paper proposes an RL-based optimal distributed control algorithm for MASs with stochastic uncertainties. By incorporating an actor-critic-identifier structure, the backstepping method, and a hybrid ETC strategy, optimality is rigorously incorporated into each step of the backstepping design, ensuring that both virtual and actual control inputs are derived as optimal solutions to their respective subsystems. The stochastic uncertainties and non-affine nonlinear faults are handled by NNs and low-pass
filter, respectively. Simulation results further demonstrate that the proposed method can substantially reduce the control update frequency while maintaining satisfactory tracking performance. Additionally, we incorporate the concept of benchmark normalization from economics to perform a novel analysis of the ETC application results. Therefore, the proposed method provides an effective and practically feasible solution for leader-following consensus control of MASs.

Future research will focus on further optimizing our distributed control algorithm. We aim to attempt to adopt a self-triggered ETC strategy to replace the existing strategy that uses the current control state, employing predictive reasoning to determine discrete time intervals and thus predict the next triggering instant of the control signal. We will also verify whether our previously proposed novel ETC result analysis strategy applies to this new framework.

\end{document}